\documentclass[11pt]{article}
\usepackage[left=1in, top=1in, right=1in, bottom=1in]{geometry}
\usepackage{indentfirst}

\usepackage{hyperref}       
\usepackage{url}            
\usepackage{booktabs}       
\usepackage{amsfonts}       
\usepackage{nicefrac}       
\usepackage{microtype}      

\usepackage[ruled]{algorithm2e} 
\usepackage{subcaption}

\SetAlFnt{\small}
\SetAlCapFnt{\small}
\SetAlCapNameFnt{\small}
\SetAlCapHSkip{0pt}
\IncMargin{-\parindent}

\usepackage{graphicx}
\usepackage{rotating}
\usepackage{amsmath}
\usepackage{dsfont}
\usepackage{float}
\usepackage{amsthm}
\graphicspath{ {./images/} }
\newtheorem{claim}{Claim}
\usepackage{xcolor}

\newtheorem{observation}{Observation}
\newtheorem{define}{Definition}
\newtheorem{theorem}{Theorem}
\newtheorem{lemma}[theorem]{Lemma}

\title{Advertising for Demographically Fair Outcomes}
\author{Lodewijk Gelauff\thanks{Management Science and Engineering, Stanford University, Stanford CA 94305. Email: {\tt lodewijk@stanford.edu}} \and Ashish Goel\thanks{Management Science and Engineering and (by courtesy) Computer Science Department, Stanford University, Stanford CA 94305. Email: {\tt ashishg@stanford.edu}} \and Kamesh Munagala\thanks{Computer Science Department, Duke University, Durham NC 27708-0129. Email: {\tt kamesh@cs.duke.edu}} \and Sravya Yandamuri\thanks{Computer Science Department, Duke University, Durham NC 27708-0129. Email: {\tt sravya.yandamuri@duke.edu}}}
\date{}

\begin{document}
\maketitle
\begin{abstract}
Online advertising on platforms such as Google or Facebook has become an indispensable outreach tool, including for applications where it is desirable to engage different demographics in an equitable fashion, such as hiring, housing, civic processes, and public health outreach efforts. Somewhat surprisingly, the existing  online advertising ecosystem provides very little support for advertising to (and recruiting) a demographically representative cohort. 

We study the problem of advertising for demographic representativeness from both an empirical and algorithmic perspective. In essence, we seek fairness in the {\em outcome} or conversions generated by the advertising campaigns.  We first present detailed empirical findings from real-world experiments for recruiting for civic processes, using which we show that methods using Facebook-inferred features are too inaccurate for achieving equity in outcomes, while targeting via custom audiences based on a list of registered voters segmented on known attributes has much superior accuracy.  

This motivates us to consider the algorithmic question of optimally segmenting the list of individuals with known attributes into a few custom campaigns and allocating budgets to them so that we cost-effectively achieve outcome parity with the population on the maximum possible number of demographics. Under the assumption that a platform can reasonably enforce proportionality in spend across demographics, we present efficient exact and approximation algorithms for this problem. We present simulation results on our  datasets to show the efficacy of these algorithms in achieving demographic parity.
\end{abstract}

\section{Introduction}
\label{sec:intro}
With the increased use of personalized ads, targeted advertising on platforms such as Google or Facebook has become an increasingly powerful tool to perform online outreach to a population. In this paper, we consider the broad problem of targeted advertising where we care about recruiting a {\em diverse} cohort that represents a target demographic mix for the purpose of civic engagement. At first blush, advertising for a representative cohort might seem like a simple problem that should follow from the vast literature on ad targeting and the wide adoption of online advertising platforms. However, the opaque nature of these platforms and the fact that they use complex machine learning algorithms are major obstacles for a good faith advertiser that wants to reach or engage a representative cohort  In this paper, we use a field study to demonstrate why this problem is challenging and motivate an algorithmic framework for achieving demographic fairness in the outcome of advertising campaigns. Our approach is primarily advertiser-centric, in that we study what an advertiser can do to target a demographically diverse mix. However, our work also provides insight into what an ad platform could do to make this problem feasible for the advertiser, without sacrificing the ability of the platform to use tools like relevance estimation and differential pricing of the same ad for different individuals.

\medskip \noindent {\bf Field Study Setup.} In Section~\ref{sec:experiment}, we present the setup and results from  running advertising campaigns on Facebook for the Participatory Budgeting elections in the cities of Durham and Greensboro in North Carolina. In Participatory Budgeting, there are a number of public projects to choose from (such as new sidewalks, park renovations, etc.). Each project has a monetary cost, and the city wants to select a set of projects given a total budget. Local community members vote directly for their preferred projects, and these votes are aggregated to decide which projects to fund. The election organizer often seeks votes from a diverse, representative cohort of residents.

The goal of the city in running these advertising campaigns was to make the demographic mix along the dimensions of race and age comparable to that of the city. In other words, we seek {\em equitable outcome} or {\em equitable conversions}: The proportion of the votes for various demographics should be close to that of the overall adult population in the city. (In our experiments, we compile the attributes of the demographics that voted by running surveys on the votes generated by these campaigns.) We use the experiments to study feasible approaches and develop an algorithmic framework. Subsequently, we simulate this framework on our data to show its efficacy for achieving demographic parity.

\medskip \noindent {\bf  Achieving High Targeting Accuracy.}
We show that a pre-requisite for achieving demographic parity is an accurate way of targeting various demographic sub-types in the population. If there were a set of campaigns that accurately target different demographics, then we can hope to allocate budgets appropriately to these campaigns to make the number of conversions proportional to the respective targeted demographic. If on the other hand, a campaign targeted at a minority demographic leads to disproportionately many conversions from the majority demographic, then even with unlimited money assigned to this campaign, the outcome cannot be equitable.

Indeed, in Section~\ref{sec:targeting}, we empirically demonstrate that Facebook's targeting methods based on multicultural affinity to ethnicity and race, or self-reported education levels have uniformly low accuracy and are not sufficient for achieving demographic parity. Even more surprisingly, we discovered that targeting a demographic minority can actually result in fewer individuals from that demographic being recruited per dollar spent compared to a wider campaign (Section~\ref{sec:pareto})!

We show a way forward to achieve high accuracy in demographic outcomes by using an approach discussed in~\cite{Korolova2}. We created a set of custom audiences, using the list of registered voters in the city, along with their (self-)identified race, age, and gender. We create campaigns where we explicitly target these voters, and show that this approach is superior in that it yields high specificity for targeting African Americans and voters by age groups. It also enables us to perform a more fine-grained study of the advertising process, and thereby work  towards achieving outcome fairness.

\medskip \noindent {\bf Balancing Fine-grained Demographics Cost-effectively.}
If we seek to approximate a random sample of the population, we need to simultaneously achieve demographic parity on combinations of several attributes (say race, age, gender, education).  Given that campaigns constructed from voter lists have sufficiently large targeting accuracy, in Section~\ref{sec:algorithms}, we turn to the question of achieving demographic parity for every combination of attributes (fine-grained demographics) cost-effectively via constructing suitably targeted campaigns. One could do this by constructing a separate custom campaign for each fine-grained demographic by using its voter list. We present arguments backed by our experiments for why this is undesirable from both privacy and cost standpoints, and why fewer campaigns with larger audiences are more cost-effective.

This motivates the algorithmic problem of segmenting the voter list into a {\em few} disjoint campaigns and allocate budgets to them so that we achieve demographic parity in outcome for as many fine grained demographics as possible.   In order to make this problem tractable, we assume the platform's allocation mechanisms are such that we can find per-demographic weights such that the spend of a campaign is split in {\em proportion to} these weights among the sub-demographics within the campaign. We argue that proportionality is an easily implementable contract between platform and advertiser. Indeed, the commonly used budget smoothing mechanism of {\em probabilistic pacing}~\cite{agwy:pacing14,bkmm:pacing17}, which uses the probability that an ad can compete for an auction to limit the amount of spending by a campaign,  implements proportionality in the large market regime. More details are in Section~\ref{sec:model}. 

Under the proportionality assumption, in Sections~\ref{sec:step} and~\ref{sec:other}, we develop efficient exact and approximation algorithms via dynamic programming for several natural objective functions related to achieving demographic parity in outcomes. Our technique involves an interesting transformation of the problem to eliminate the non-linearity arising due to the presence of proportions in the objective. In Section~\ref{sec:online}, we present heuristics for implementing these algorithms online. 

\medskip
It is important to note that the proportionality assumption does not require the platform to assign identical relevance weights or charge identical prices for all demographics (or all individuals in the same demographic).  In fact, the platform can set arbitrary prices for different individuals, and arbitrarily exclude individuals from seeing an ad, and still satisfy proportionality, provided that it uses probabilistic pacing and the relevance score/filtering criteria for an individual and an ad depend only on the characteristics of the individual and the ad, and not on which targeting criteria were used. 

Finally, in Section~\ref{sec:simulate}, we run the optimization algorithm on our data (after  inferring conversion rates for the demographics), and it provides an upper bound on how well we could have balanced on attribute combinations (such as race, age, and gender) if such balancing was our sole objective as opposed to the experiments presented above. We compute such an empirical upper bound, and show that it is indeed possible to achieve approximately fine-grained parity on most race-age-gender combinations with modest cost per vote, and reasonably large audience size per campaign.

\medskip \noindent{\bf Summary.} In summary, our work can be viewed as both a positive as well as a negative result. On the positive side, it demonstrates the power of advertising to curated lists of potential participants, and leads to principled optimization approaches which only require minor (if any) modifications to the ML routines used by the platform for targeting and pricing. On the negative side, it shows that when such lists are not available, it is not possible to treat the platform's behavior as a black box, and we would need to re-engineer these systems to take demographic diversity into account. 

It is important to note that while we point out several deficiencies in Facebook's ad platform {\em for our specific problem}, the advertising eco-system is complex, and it would not be appropriate to ascribe any wrong-doing to Facebook on the basis of our work. For example, it is entirely conceivable that Facebook imposes an extra surcharge for targeting small populations or racial characteristics to avoid privacy-loss and exploitation of vulnerable segments of society.

We finally note that though we have used Participatory Budgeting elections as the main application domain to explore the problem of advertising to diverse cohorts, our techniques and insights may be applicable more broadly in contexts where it is desirable to perform outreach to various demographics in an equitable fashion, for instance, in hiring, housing, and public health outreach efforts.

\medskip \noindent{\bf Roadmap.} We describe  our experimental setup in Section~\ref{sec:experiment}, and present our results for targeting in Section~\ref{sec:targeting}. Section~\ref{sec:algorithms} presents the  algorithmic formulation of the problem of constructing campaigns and budgets for fine-grained demographic balance, the rationale for the formulation, efficient algorithms, and an empirical evaluation on our dataset. Section~\ref{sec:discuss} provides a list of open problems, including a discussion of normative and ethical concerns.

\subsection{Related Work}
\noindent {\bf Online advertising.} Online advertising  -- including the problems of targeting, allocation, auctioning, and pricing --  has been widely studied  in the context of maximizing the {\em expected return} which can be measured as either impressions, clicks, or conversions (i.e. when the user takes the desired action such as purchase a product or complete a survey). The literature is too vast to cite here, but there has been significant work in auction design~\cite{agm:auction06,EOS07,myerson,AtheyE},  online allocations in advertising~\cite{Devanur,Devanur2,Mehta}, and design of advertising exchanges~\cite{Balsiero,GoelLMNP}. At a high level, in the former, the platform performs dynamic allocations of ads to publishers given long-term budgets for advertisers, while in the latter, the platform acts as an intermediary given reservation prices from publishers and bids from advertisers, both of which are dynamic.  However, the problem of dynamically selecting a representative cohort that actually converts has not been studied in either of these settings.

 \medskip \noindent {\bf Fair Advertising.} There has been very recent work that addresses the problem of demographic parity from the platform's perspective. For instance, the work of~\cite{CelisV} considers the problem of auction design taking into account demographic constraints, and the work of~\cite{NasrT} considers the problem of dynamically adjusting the bid of an advertiser via an MDP formulation to achieve parity among fixed categories (such as number of male versus female impressions). While these works consider the problem from the platform's perspective and hence focus on equity in allocation of budgets or impressions, we focus on equity in outcomes, which is the natural next step given a platform can reasonably balance budgets or impressions. This leads to interesting algorithmic as well as empirical insights that are different from previous work. 

A high level challenge with advertising for equitable outcomes is that platforms don't always expose a full range of demographic targeting filters for various legal and ethical reasons -- a tool that can be used for promoting diversity by allowing an advertiser to reach specific audiences can also be used for discriminating against protected groups. In some contexts such as housing or lending, such discrimination is explicitly prohibited by law.  Further,  the recent work of~\cite{Korolova2,ekksk:ads18} shows that Facebook's existing ad targeting mechanisms can lead to unfair outcomes because they rely on complex machine learning models and among other things, they point out the role of complex machine learning models that cause ads indistinguishable to humans to effectively target very different demographic groups. 
In contrast, our goal is less about understanding {\em why} disparate targeting or conversions occur, but rather about {\em how} their effects can be mitigated from the advertiser's perspective.

 \medskip \noindent {\bf  Fairness in Automated Decision Making.} . 
 Fairness is a nuanced topic, and our experiments view it as demographic parity (equity of outcomes). Recent work~\cite{DworkI,dg:fair18}  has discussed a range of fairness conditions for automated decision making, and 
 show that how even when an individual decision can seem fair, it may have an unfair effect when multiple algorithms and decisions collide. In particular, this work distinguishes between individual (equal characteristics = equal probabilities) and group-fairness (demographic groups get similar treatments).  Some preliminary progress has recently been made in incorporating notions of individual fairness into advertising models~\cite{kkr:pif2020}, but incorporating these notions into our advertiser-centric approach remains an interesting open direction. The trade-off between fairness and efficiency has been extensively studied in classification (see for instance~\cite{Parkes}) and our work performs a similar analysis in the context of advertising.  

\medskip \noindent {\bf Team formation.} A different line of work that we draw inspiration from is the formation of {\em diverse teams}~\cite{Liemhetcharat2014,Lappas2009,LiShan2010,Rahman2019,sandholm99}. In these problems, we are given a set of people with skills in different dimensions, and the goal is to choose one team or multiple teams, where each skill is represented. Such models have applications in the design of crowd sourcing platforms. A stylized view of our problem is as team formation where we view each set of advertising campaigns as a team that yields a demographic mix. Though creating such campaigns is superficially related to diverse team formation, the crucial difference is that our objective is not utility-maximization, but instead balancing the demographic mix that is obtained as a result of the campaigns.

 \medskip \noindent {\bf Collective Decision Making.} There has been a recent resurgence in the field of social choice, the science of making collective decisions. Some practical examples are Participatory Budgeting~\cite{PBP}, the MIT moral machine~\cite{MITMORAL}, sortition~\cite{BenadeGP19}, and delegative democracy~\cite{gkmp:liquid18}.  Our premise in this paper is that for such decision making processes to work effectively, it is important to recruit participants that are representative of the underlying population. 

We note that the problem we study is quite different from that of running opinion polls, where one can re-weight the results of the poll to correct for any demographic disparity in participation~\cite{gs:rds10}; this cannot be done in civic processes due to ethical constraints.

\section{Experimental Setup with Participatory Budgeting Elections}
\label{sec:experiment}
We now describe our  experimental setup and results for advertising for a demographic mix in the context of the Participatory Budgeting elections in Durham and Greensboro, North Carolina. We show example ads in Fig~\ref{fig:adexample}.

\begin{figure}[htbp]
\begin{center}
\parbox{0.45\textwidth}{
\begin{center}
\includegraphics[width = 0.3\textwidth]{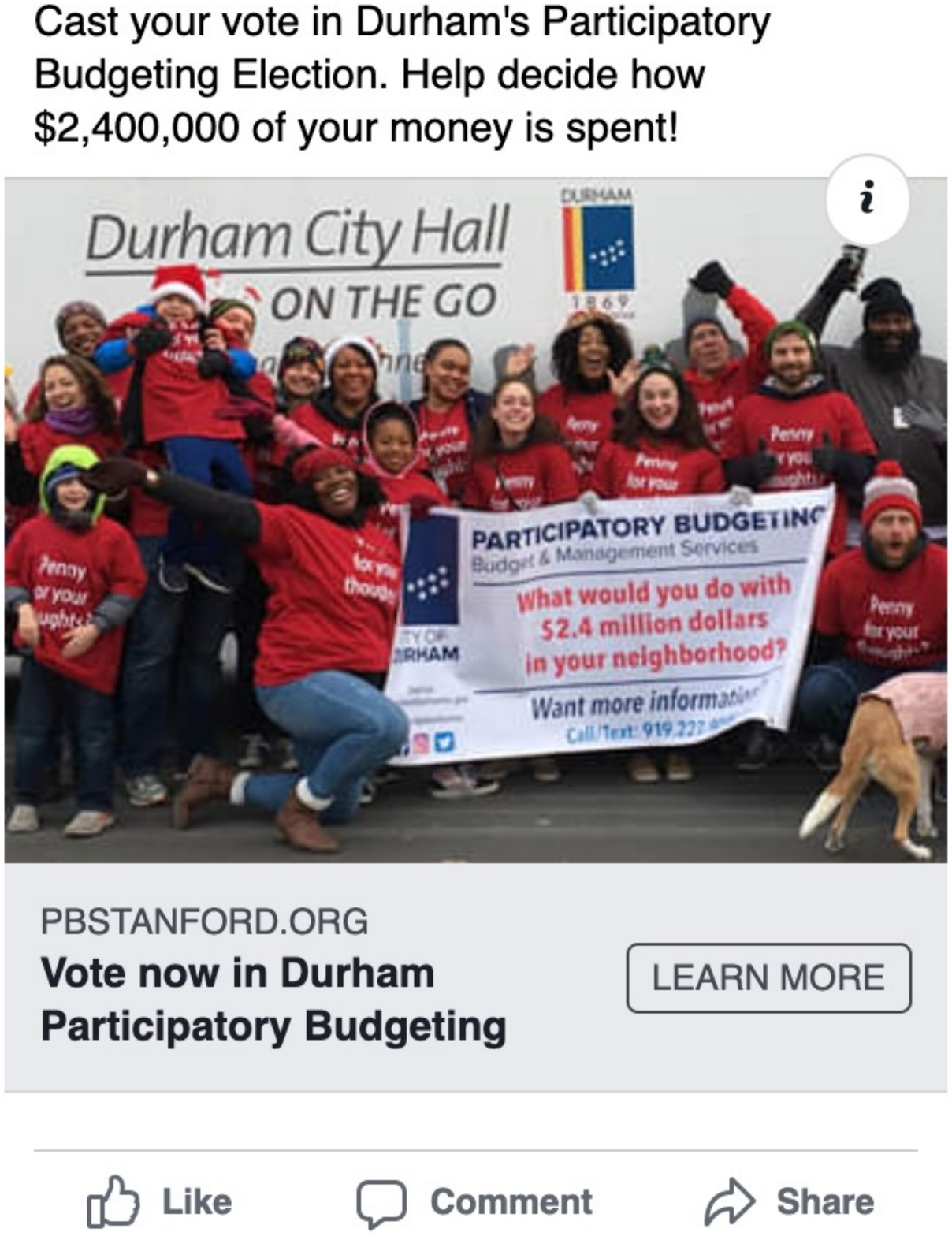}
\end{center}
}
\hspace{0.2in}
\parbox{0.45\textwidth}{
\begin{center}
\includegraphics[width = 0.3\textwidth]{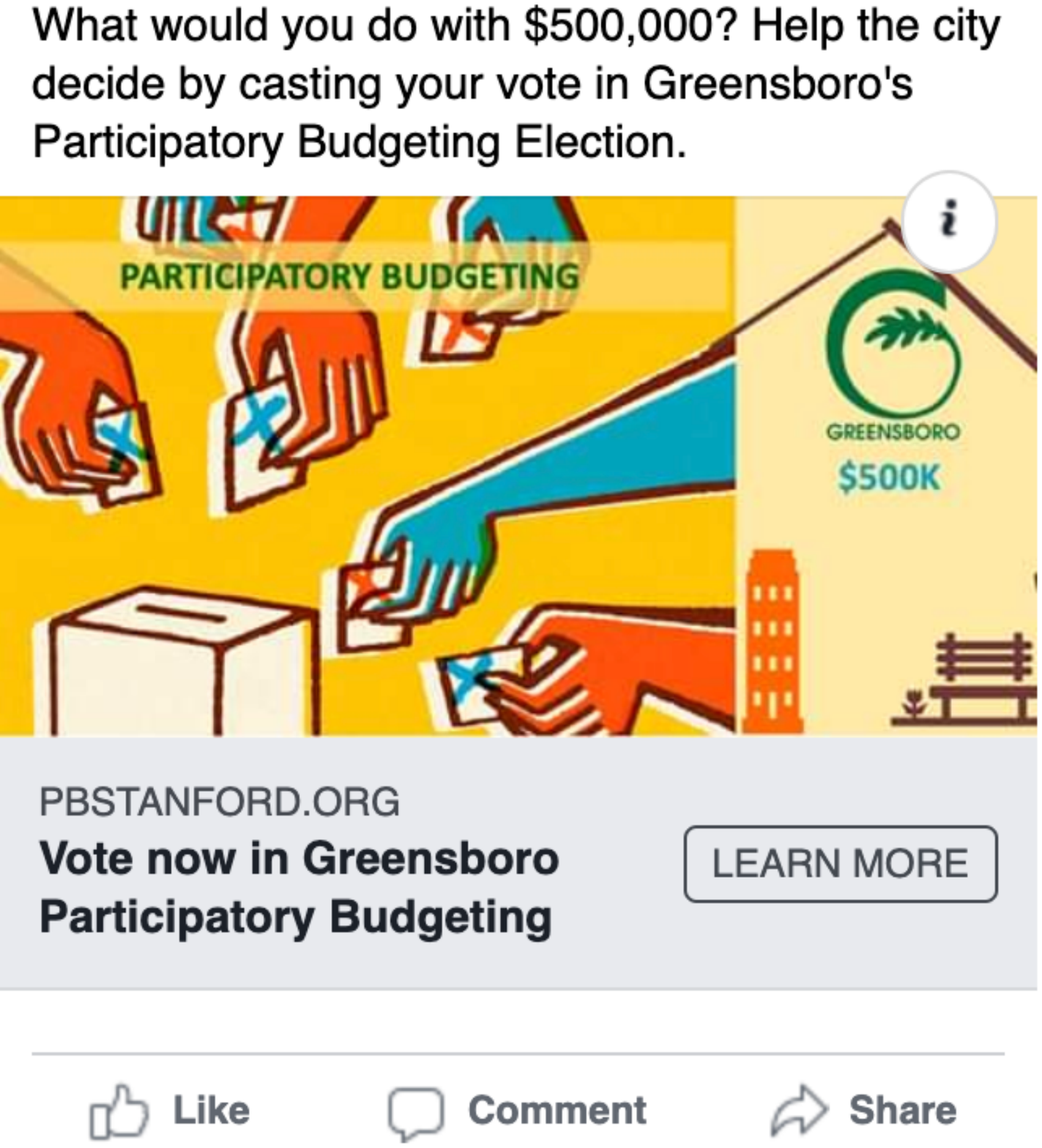}
\end{center}
}
\caption{\label{fig:adexample} Example Participatory Budgeting Facebook ads in Durham (left) and Greensboro (right).}
\end{center}
\end{figure}

Each advertising campaign on Facebook specifies an ad creative, an audience, a budget and a bidding strategy. For each city, we created a set of parallel campaigns  with an identical ad creative  that invites city residents to vote in the Participatory Budgeting election that was organized by the city. Each ad campaign is targeted to a specific demographic via a targeting criterion; we preset details of these campaigns and criteria in Appendix~\ref{app:details}. The bidding strategy is always set to optimize for number of clicks, as recommended by Facebook. If a user clicks on the ad, they are taken to the election landing page where they can choose their district, register, and vote. If a voter completes the voting process, they are then prompted to complete a satisfaction and demographic survey. A parameter in the URL tracks which campaign they arrived through, which allows us to establish the demographic characteristics of residents that arrived through each of the campaigns and completed the survey.

We use \textit{equitable outcome} (proportionate completion of surveys across demographics) as our desired objective because a proportionate ballot completion across demographics is what makes the process most representative for the city population. We see a high survey completion rate of 75\% or higher across demographics.

Our end-goal is to develop a framework for constructing a set of campaigns targeted at various demographic segments and allocate budgets between them so that the resulting conversions are proportional to the size of those demographics.  Since the goals of achieving equity of outcome and performing experiments to unearth feasible approaches to doing it  conflict with each other, we focus on the latter goal in our experiments. In Section~\ref{sec:simulate}, we simulate the algorithmic framework derived from the experiments on our data to show how equity of outcomes could have been achieved had it been our primary goal.

\section{Achieving High Targeting Accuracy}
\label{sec:targeting}
In this section, we present experiments that compare different targeting schemes in terms of their effectiveness in targeting the relevant demographics. We do so by measuring {\em specificity} of the corresponding campaign, defined as follows.

\begin{define}
Let $\alpha_i$ denote the fraction of demographic $i$ in the population. The {\em specificity} of a campaign $j$ for demographic $i$, denoted $\beta_{ij}$ is fraction of outcomes of that campaign that were from  demographic $i$. We say the specificity is ``high" if $\beta_{ij} \ge \alpha_i$ and ``low" otherwise.
\end{define}

High specificity is a pre-requisite for equity as shown by the following straightforward lemma.

\begin{lemma}
Given a set $S$ of campaigns and demographic $i$, if we have low specificity, $\beta_{ij} < \alpha_i$ for all $j \in S$, then the proportion of outcome generated for demographic $i$ by {\em any} budget allocation to $S$ is less than $\alpha_i$, so that we cannot achieve equitable outcomes.
\end{lemma}

For instance, if any campaign, including that targeted at African Americans, only leads to at most $20\%$ conversions from African Americans, while they represent $40\%$ of the population, we cannot hope to achieve demographic parity since in any budget allocation, at most $20\%$ of the conversions will be African American. We therefore need specificity of a targeted campaign to be at least the proportion of that demographic in the population. However, the ``high specificity" threshold of $\alpha_i$ is necessary but need not be sufficient to construct campaigns to achieve demographic parity, since we need to achieve parity across several demographics simultaneously.  Ideally we would like specificity of a targeted campaign to be close to $1$. Surprisingly, we show that several targeting approaches fail even the baseline ``high specificity" threshold.

\subsection{Facebook-inferred Targeting: Low Specificity and Pareto Sub-Optimality} 
\label{sec:pareto}
As described in detail in Appendix~\ref{app:durham}, we started in Durham by naively targeting through Facebook-inferred characteristics related to race and ethnicity. For the African American campaign, the campaign was targeted towards those in Durham with an interest in African American culture and behavior identified as African American multi-cultural affinity (AfrAm1); similarly for Hispanic culture and affinity (Hisp1). In addition we ran one campaign excluding everyone with any college experience (Educ), and one campaign targeted to the general population (General). 

In Table~\ref{fig:specificity}, we show the specificity of the (Educ), (AfrAm1) and (Hisp1) campaigns for that respective demographic in the second, third, and fourth columns.   We observed low specificity from the Facebook targeted campaigns for race and education level. Note that the specificity for the African American campaign (AfrAm1) is $20\%$, compared to the proportion of $37\%$ in the population. The specificity is so low that the best solution to obtain a race-equitable outcome, allocating the entire budget to the (AfrAm1) campaign, had a detrimental effect on the overall vote proportions, which meant it is not possible to achieve demographic parity for African Americans via such campaigns.  In fact, the organic efforts of the City of Durham achieved a better proportion of African Americans.  Note that though the specificity of the (AfrAm1) campaign was poor, it was more specific than the (General) campaign with respect to this demographic, showing that just using the (General) campaign would not have achieved parity.  We summarize our observation as:

\begin{table}[htbp] 
\centering
 \begin{tabular}{|r | c  c  c  c  c | c  c  c  c |} 
 \hline
    & \multicolumn{5}{c|}{Durham} & \multicolumn{4}{c|} {Greensboro} \\
     Demographic & \begin{sideways} Educ \end{sideways} & \begin{sideways} Afram1 \end{sideways} & \begin{sideways} Hisp1 \end{sideways} & \begin{sideways} AfrAmLAL \end{sideways} & \begin{sideways} AfrAmC \end{sideways}  & \begin{sideways} YoungC \end{sideways} & \begin{sideways} MediumC \end{sideways} & \begin{sideways} OldC \end{sideways} & \begin{sideways} AfrAmC \end{sideways} \\
     \hline
    Population Fraction $\alpha$ & 0.13 & 0.37 &  0.14 & 0.37 & 0.37 & 0.55 & 0.30 & 0.15 & 0.41 \\
    \hline
     Achieved Specificity $\beta$ & {\color{red} \bf 0.12} & {\color{red} \bf 0.21} & {\color{red} \bf 0.12} & {\color{red} \bf 0.21} & {\color{blue} \bf 0.69}  & {\color{blue} \bf 1.0} & {\color{blue} \bf 0.93} & {\color{blue} \bf 0.91} & {\color{blue} \bf 0.70} \\
     \hline
     95\% CI Lower bound & 0.06 & 0.16 & 0.05  & 0.08 & 0.57 & 0.93 & 0.82 & 0.76 & 0.6 \\
     95\% CI Upper bound & 0.2 & 0.27 & 0.22 & 0.39 & 0.79 & 1 & 0.98 & 0.99 & 0.81 \\
     \hline
     \end{tabular}
     \caption{\label{fig:specificity} Specificity of targeted ad campaigns. $\alpha$ is the fraction of the targeted demographic in the population, $\beta$ is the mean specificity for that demographic via the campaign, and the $95\%$ CI assumes the counts follow a Binomial distribution. The high specificity voter-list based custom campaigns ($\beta \ge \alpha$) are in {\color{blue} \bf blue} and the low specificity campaigns ($\beta < \alpha$) are in {\color{red} \bf red}.}
     \end{table}

\begin{observation}
The specificity of the African American campaigns is low when we target by Facebook-inferred criteria such as multi-cultural affinity or interest in culture. The specificity of Facebook-inferred Education criteria is also low 
\end{observation}

\paragraph{Pareto non-optimality of Targeting.} In addition to having low specificity, these targeted campaigns may not be cost-effective even for the demographic they are targeting. The first three columns of Fig~\ref{fig:durhamVotesDollar} show the cumulative African American votes per dollar for two of the Durham 2019 campaigns. We make a surprising observation:

\begin{observation}
The General campaign (General) has a higher rate of African American votes per dollar than the African American  culture/multi-cultural affinity campaigns (AfrAm1 and AfrAm2).
\end{observation}

\subsection{Targeting via Voter List-based Custom Audiences: High Specificity}
\label{sec:greensboro}
The low specificity and high cost of the Facebook-inferred targeting criteria led us to targeting with custom audiences created from the North Carolina list of registered voters. In order to specify a custom audience, we take the publicly available voter registration database of North Carolina, filter for the desired characteristics and upload a list of people with personally identifiable information. Facebook matches a set of profiles to the list, which then becomes the audience for that campaign. 

\begin{figure}[htbp]
\begin{center}
\parbox{.46\linewidth}{
\includegraphics[width = 0.45\textwidth]{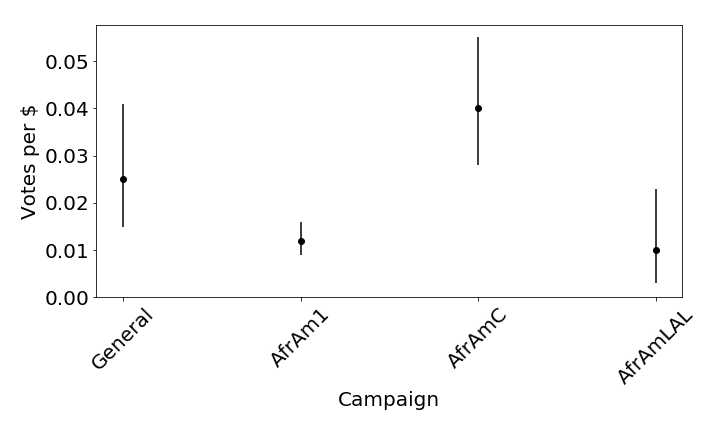}
\subcaption{\label{fig:durhamVotesDollar} Number of African American votes achieved per dollar spent on the various campaigns in Durham. Note that the targeted (AfrAm1) campaign is worse for this demographic than (General).}
}
\hspace{0.2in}
\parbox{.46\linewidth}{
\includegraphics[width = 0.45\textwidth]{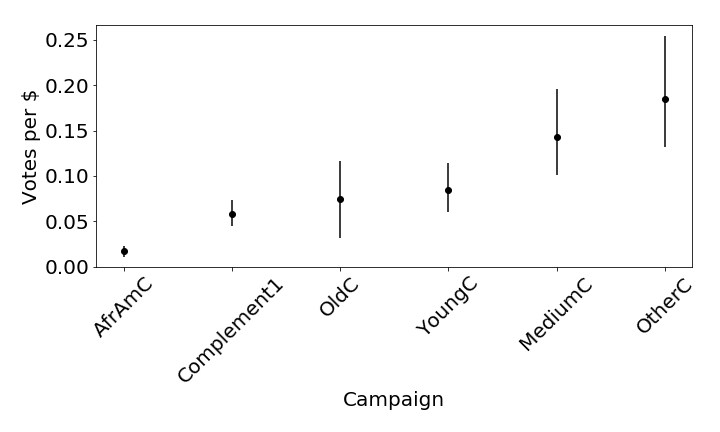}
\subcaption{\label{fig:greensboroVotesDollar} Votes of respective demographic per dollar spent on entire campaign for Greensboro 2019 Custom Audience and complement Campaigns. For the ``complement" campaign, we consider all its votes.}
}
\caption{The vertical bar shows 95\% CI around the mean (solid dot) assuming a Poisson process.}
\end{center}
\end{figure}

The sixth column of Table~\ref{fig:specificity} contains the specificity for African Americans of the custom campaign (AfrAmC) in Durham NC created with the list of African American voters. This has high specificity $\beta = 70\%$, relative to their $\alpha = 37\%$. Similarly, in Fig~\ref{fig:durhamVotesDollar}, this campaign has high votes per dollar and Pareto-dominates the (General) campaign. To validate this conclusion further, in the next experiment in Greensboro (Appendix~\ref{app:greensboro}),  we used 5 custom campaigns created using voter lists along two demographic axes: age -- (YoungC) aged 18-44, (MiddleC) aged 45-64, and (OldC) aged 65 and up; and whether someone was African American (AfrAmC) or not (OtherC). A sixth ``complement" campaign, (Complement1), was created targeting the whole city except those people part of any of the other voter-list based campaigns, in order not to fully exclude anyone from the advertising process. The final four columns of  Table~\ref{fig:specificity} shows the high specificity of the Greensboro age and race custom campaigns for the respective demographic, which validates the conclusion from Durham and is summarized below. 

\begin{observation}
Campaigns that create custom audiences using known lists of people within a demographic have high specificity for age and race criteria.
\end{observation}

\subsection{Other Experimental Results}
\label{sec:misc}
We present other experimental results in Appendix~\ref{app:additional}. These are relevant to the model developed in the next section. We describe two that pertain to specificity and cost next.

\medskip \noindent {\bf Lookalike Campaigns have Low Specificity.} In an attempt to also efficiently target people that were not registered voters, we then created campaigns with "lookalike" audiences based on the respective filtered voter lists; these campaigns targeted users that Facebook thinks look like the voters in our custom audience.  These campaigns are described in detail in Appendix~\ref{app:details}. The lookalike campaign based on the African American voter list in Durham (AfrAmLAL) resulted in poor specificity (see fifth column in Table~\ref{fig:specificity}) and low votes per dollar (Fig~\ref{fig:durhamVotesDollar}). We present results for lookalike campaigns in Greensboro NC in Appendix~\ref{app:lookalike} showing the poor specificity of this feature even along the dimensions of age and gender.

\medskip \noindent {\bf Conversion Rates.} Though the voter list campaigns are highly specific, they can have large variation in conversion rates that we show in Fig~\ref{fig:greensboroVotesDollar}. Note that the conversion rates between the (AfrAmC) and the (OtherC) campaign differ by a factor of more than six.

\medskip \noindent {\bf Other Empirical Results.} In Appendix~\ref{sec:k_campaigns}, we show the variation in cost per conversion as a function of campaign audience size; we use this result in building our model in Section~\ref{sec:model}. In Appendix~\ref{app:budget}, we analyze how the campaign budget is spread along sub-demographics within a voter list; we use this result in our simulation study in Section~\ref{sec:simulate}. Finally, in Appendix~\ref{app:gender}, we analyze equity in outcomes along demographics that are orthogonal to the targeted demographics; we use this result in building our model in Section~\ref{sec:model}.

\section{Campaign Design for Fine-Grained Demographic Balancing}
\label{sec:algorithms}
\label{sec:joint}
The previous section showed that campaigns based on targeting lists of individuals have high enough specificity so that they can used in our application of demographic parity. Given this ability to accurately target demographics, we now study the question of how to construct cost-effective campaigns that can simultaneously achieve equality of outcome on fine-grained demographics such as all possible (race, age, gender) combinations. We formulate this as an optimization problem with constraints and assumptions that we motivate next.

\subsection{Modeling and Assumptions}
\label{sec:model}
Formally, we are given a set $D$ of $n$ demographics, where demographic $i$ comprises a fraction $\alpha_i\in[0,1]$ of the population. We assume these are disjoint, so that $\sum_i \alpha_i = 1$. These demographics could be all possible combinations of race, age, and gender.  For demographic $i \in D$, let $\phi_i$ denote the number of conversions per dollar for this demographic, had we constructed a campaign targeting the voter list restricted to this demographic.  A daily budget of $B$ can be split between the campaigns. 

\medskip \noindent {\bf Number of Campaigns.} First note that we cannot simply construct one campaign with the entire voter list. Even in the ideal case that the platform partitions the budget equitably, as mentioned before, different demographics can have large variation in conversion rates $\phi_i$ that we show in Fig~\ref{fig:greensboroVotesDollar}. Note that the conversion rates between the (AfrAmC) and the (OtherC) campaign differ by a factor of more than six. Further, splitting voter lists by one criterion such as age may produce inequitable outcomes along another dimension such as gender; see Appendix~\ref{app:gender}.

We could go the other extreme and create one custom campaign for each of the $n$ demographics, assigning budgets to them achieve parity in conversions. However, this will lead to too many campaigns each of which is constructed using a small voter list and small budget. This approach can lead to too few matches with Facebook's population, thereby violating privacy of individual users. In Appendix~\ref{sec:k_campaigns}, we present experimental evidence that splitting the voter list of a campaign into smaller parts with proportionally lower budget has {\bf higher cost per conversion} and higher variance in this value compared to larger campaign, making the smaller campaigns much more expensive and harder to learn parameters of. This effect is likely by design -- smooth delivery algorithms~\cite{Borgs,probpacing} can cause high variance effects for small budgets; further, the platform can conceivably put a premium on small campaigns to penalize discriminatory behavior. 

To make the campaigns cost-effective and privacy preserving, we therefore enforce a constraint $k$ on the total number of campaigns we can create. Continuing our formalism, our goal is to partition the $n$ demographics into $k$ disjoint campaigns, each targeting the relevant subset of the demographics (in our case, via targeting their voter list). A campaign $j$ is represented as a tuple ($B_j$, $S_j$), where $B_j$ is the campaign's daily budget and $S_j \subseteq D$ is the subset of demographics targeted by that campaign. Note that $\sum_j B_j \le B$.  We assume the $S_j$ are disjoint and $\cup_j S_j \subseteq D$. This allows for the possibility of omitting some demographics from the campaigns. 

\medskip \noindent {\bf Proportionality Assumption.} 
If a campaign spans more than one demographic, we need an assumption on how the budget of the campaign is split among the constituent demographics in order to make the algorithm design and parameter learning tractable. Towards this end, we make a proportionality assumption: Given a campaign $j$ with demographics $S_j$, we assume its budget is allocated to demographic $i \in S_j$ in proportion to a known parameter, $q_i$, that depends only on this demographic and not on the other demographics in set $S_j$. Under this assumption, the expected number of votes per dollar for demographic $i \in S_j$ (that is, the conversion rate) is therefore:
\begin{equation}
\label{eq:eta2}
    \eta_{ij} = \phi_i \frac{q_i}{\sum_{\ell \in S_j} q_{\ell}} 
\end{equation} 

We now argue that proportionality is a simple and easily implementable contract between the platform and the advertiser. If the platform uses {\em probabilistic pacing}~\cite{agwy:pacing14} to smooth the daily budget $B$, we can idealize its behavior as follows. Suppose with infinite campaign budget, an individual's pricing and relevance estimation for the ad does not depend on which other demographics are in the campaign. In this infinite budget case, let $q_i$ be the money that would have been spent on demographic $i$.  Further suppose the audience of the campaign arrives in random order uniformly over time. Then the total spend on campaign $j$ in the infinite budget setting is $\hat{B} = \sum_{\ell \in S_j} q_{\ell}$. Under probabilistic pacing, the platform considers each impression independently with probability $B/\hat{B}$. Clearly, the expected number of conversions for any subset of impressions is proportional, so that those for demographic $i$ will be $q_i B/\hat{B}$, hence showing proportionality. Note that proportionality is more general than the assumption that spend should be proportional to audience size since the corresponding demographics could have very different pricing characteristics.  Indeed, this assumption allows the platform considerable latitude in developing relevance and pricing criteria, and hence can be thought of as a "minimal contract" between an ad platform and an advertiser. 


\subsection{Objective Function}
Under the proportionality assumption (Eq~(\ref{eq:eta2}), we seek to construct at most $k$ disjoint campaigns and allocate budgets to them, so that the total budget is at most the daily budget $B$. Suppose in the outcome, the fraction of conversions for demographic $i$ is $\rho_i \alpha_i$ for $\rho_i \ge 0$. We  optimize an objective function of the form $\sum_{i \in D} f(\rho_i)$ corresponding to demographic parity. 

We consider the {\bf $\gamma$-step function} objective where we fix a parameter $\gamma \in [0,1]$, and define $f(\rho_i) = 1$ if $\rho_i \ge \gamma$ and $0$ otherwise. This corresponds to maximizing the number of demographics whose representation in the votes is at least a $\gamma$ factor of their representation in the population. 

\medskip \noindent {\bf Mathematical Program.} Since we are interested in proportions in outcomes, we can phrase the problem somewhat differently. Our goal is to construct at most $k$ disjoint campaigns and allocate budgets $b_j$ to them, where $b_{j}$ denotes the campaign budget scaled down by the overall number of conversions. Recall the definition of $\eta_{ij}$ from the proportionality assumption, Eq~(\ref{eq:eta2}). Suppose the fraction of conversions (votes) for demographic $i$ is $\rho_i \alpha_i$ for $\rho_i \ge 0$. We want: (i) to enforce the constraint that the expected number of conversions is exactly one:
$$ \sum_{i \in D} \rho_i \alpha_i = \sum_{j = 1, 2, \ldots, k} \sum_{i \in S_j} b_j \eta_{ij}  = 1 $$
and (ii) to optimize an objective function corresponding to demographic parity in conversions:
\begin{equation}
    \label{eq:obj}
    \mbox{Maximize} \sum_{i \in D} f(\rho_i) 
\end{equation} 

The $b_j$ values above are the budgets needed to {\em generate one vote} in expectation. Given a daily budget $B$, we will simply allocate it to the campaigns in proportion to the $b_j$ values, in other words, the daily budget of campaign $j$ is $B_j = \frac{b_j}{\sum_{j'} b_{j'}} B$. This preserves the ratios $\rho_i$ for all $i \in D$.

\medskip
In the integer non-linear program below,  $\ell$ denote campaigns. $x_{i \ell}$ is the indicator variable whether demographic $i$ successfully achieves its $\gamma \alpha_i$ share from campaign $\ell$, and $b_{\ell}$ is as defined above.  The objective captures the number of successful demographics. The first constraint captures that each successful demographic was assigned to at most one campaign; the second constraint, when $x_{i \ell} = 1$ says that parity was achieved for demographic $i$; and the third constraint encodes that $\{b_{\ell}\}$ generate one vote overall.  
\[
\mbox{Maximize} \sum_{i \in D} \sum_{\ell=1}^k x_{i \ell} \]
\[
\begin{array}{rcll}
\sum_{\ell=1}^k x_{i \ell} & \le & 1 & \forall i \in D \\
b_{\ell} \left(\phi_i q_i\right)/\left(\sum_{j \in D} q_j x_{j \ell}\right) & \ge & \gamma \cdot \alpha_i x_{i\ell}   & \forall i \in D, \ell = 1,2, \ldots, k \\
\sum_{\ell = 1}^k b_{\ell} \left(\left(\sum_{i \in D} \phi_i q_i x_{i \ell}\right)/\left(\sum_{i \in D}  q_i x_{i \ell}\right) \right)& = & 1 \\
x_{i \ell}   \in   \{0,1\} & \mbox{and} & b_{\ell}  \ge  0 & \forall i \in D, \ell = 1,2, \ldots, k
\end{array}
\]

\subsection{Efficient Algorithm} 
\label{sec:step}
Despite the non-linearity in the above formulation, our main algorithmic contribution is the following surprising theorem. 

\begin{theorem}
\label{thm:main}
The joint campaign design and budget allocation problem assuming a $\gamma$-step function objective can be exactly solved in $O(k n^5)$ time by dynamic programming.
\end{theorem}

The proof first involves a non-trivial change of variables to eliminate the non-linearity. We subsequently show that the problem is amenable to dynamic programming. 


\medskip \noindent {\bf Transformed Formulation.}  We first use Eq~(\ref{eq:eta2}) to transform the problem and get rid of the ratios in $\eta_{ij}$,  so that it is amenable to efficient computation. We set 
$$w_i = q_i \phi_i; \qquad Y_j = \frac{b_j}{\sum_{\ell \in S_j} q_{\ell}}; \qquad\beta_i = \frac{\alpha_i}{q_i \phi_i}$$ 
where $\phi_i$ and $q_i$ are the conversion rates and spend proportionality factors respectively for demographic $i$.

Let $\sigma(i) = j$ if $i \in S_j$. Then, the number of conversions for demographic $i \in S_j$ is 
$$ \rho_i \alpha_i = b_j \eta_{ij} = Y_{\sigma(i)} w_i \qquad \implies \qquad \rho_i = \frac{Y_{\sigma(i)}}{\beta_i}.$$

Therefore, the input is $(w_i, \beta_i)$ for each demographic, and the output, for each campaign, is the number $Y_j$ and set of demographics $S_j$. The constraint of one vote in expectation becomes
\begin{equation}
    \label{eq:exactlyone}
 \sum_j Y_j \sum_{i \in S_j} w_i = 1
\end{equation}
Given the output $Y_j$ and $S_j$, we can calculate $b_j$ as $b_j = Y_j \sum_{\ell \in S_j} q_{\ell}$. The tuple $(S_j, Y_j)$ therefore completely describes campaign $j$. In order to capture demographics not assigned to any campaign, we create a dummy campaign with $b_0 = Y_0 = 0$, and $S_0$ as the set of demographics assigned to it. Therefore, if $\sigma(i) = 0$, demographic $i$ is not assigned to any campaign. 

\medskip
With the above transformation, Objective~(\ref{eq:obj}) becomes:
$$ \mbox{Maximize} \sum_{i \in D} f\left(\frac{Y_{\sigma(i)}}{\beta_i} \right)$$

For the $\gamma$-step objective described above, since $f$ is monotonically non-decreasing, we can replace Constraint~(\ref{eq:exactlyone}) with the constraint that there is {\em at most} one vote in expectation; any solution satisfying the latter constraint can be converted to satisfy the former constraint by scaling up the $Y_j$'s, while not decreasing the objective. Therefore, for any monotone function $f$ (see Section~\ref{sec:other} for more on monotone functions), the following constraint suffices:
\begin{equation}
    \label{eq:atmostone}
     \sum_j Y_j \sum_{i \in S_j} w_i \le 1
\end{equation}

%

\medskip \noindent {\bf Dynamic Programming.}  Recall that for parameter $\gamma$, $f(x) = 1$ if $x \ge \gamma$, and $0$ otherwise. Let $\hat{\beta_i} = \gamma \beta_i$. Then the objective is:
$$\mbox{Maximize} \sum_{i \in D} \mathbf{1}_{Y_{\sigma(i)} \ge \hat{\beta_i}}$$

Given any setting of $Y_1 \le Y_2 \le \cdots \le Y_k$, let $\theta_i = \mbox{argmin}_j \{ Y_j \ge \hat{\beta_i}\}$. Note that if $i \in D$ is assigned to a campaign (and hence contributes to the objective), it must be assigned to a campaign with $Y_j \ge \hat{\beta_i}$. Regardless of which campaign it is assigned to, the contribution to the objective is $1$; however, assigning to a campaign with smaller $Y_j$ only makes Constraint (\ref{eq:atmostone}) more feasible. Therefore, we have the following claim:

\begin{claim}
\label{claim:ahead}
If $i \in D$ is assigned to a campaign (and hence contributes to the objective), it must be assigned to campaign $\theta_i$, that is, $\sigma(i) = \theta_i$.
\end{claim}

Therefore, if we sort $\hat{\beta_i}$ in increasing order (and number them as $1,2, \ldots$ in that order), each optimal $Y_j$ will correspond to one of these values, and all $i$ whose $\hat{\beta_i}$ lie between $Y_{j-1}$ and $Y_j$ will either lie in $S_j$ or not be assigned to any campaign.

\medskip
Given the above observation, let $V_j(i,m)$ denote the best set of $j$ campaigns, where $Y_j = \hat{\beta_i}$, that satisfy at least $m$ of the first $i$ demographics and that minimizes the expected number of conversions. The final goal is find the maximum $m$ so that $\min_{i=1}^n V_j(i,m) \le 1$.

For $j > 1$, we have the recurrence:
$$ V_j(i,m) = \max_{1 \le \ell < i, m' \le m} \left\{ V_{j-1}(\ell,m-m') + K(\ell,i,m') \right\}$$
For $j = 1$, we have $V_1(i,m) = K(1,i,m)$.
Here, $K(\ell,i,m)$ is the minimum expected number of conversions needed to satisfy at least $m$ of the demographics $\{\ell+1, \ldots, i\}$ by assigning to a single campaign with $Y = \hat{\beta}_i$. In other words, $K(\ell,i,m)$ is the solution to a {\sc Knapsack} problem. Here, item $i$ having unit profit and size $w_c \hat{\beta_i}$. The goal is to find the smallest knapsack size to achieve profit at least $m$. This can be exactly solved in $O(n^2)$ time by dynamic programming, and by a straightforward calculation of running time,  implies Theorem~\ref{thm:main}.

\subsection{Other Objectives}
\label{sec:other}
The basic framework described above extends to more general objective functions $f(\rho)$. We consider two general functions $f$ that we motivate next.

\medskip \noindent{\bf Monotone functions.} The $\gamma$-step  function does not give any credit to demographics that are even slightly less satisfied than $\gamma \alpha_i$. We can generalize this to smoother functions that penalizes demographics whose fraction in the conversions is less than the fraction in the population, but there is no penalty for a demographic being over-represented. In other words, for monotone functions, we have $f(\rho_i)$ is monotonically non-decreasing for $\rho_i \in [0,1]$ with $f(0) = 0$ and $f(\rho_i) = 1$ if $\rho_i \ge 1$. This not only generalizes $\gamma$-step functions, but also captures functions of the form $f(\rho_i) = \min(1,\rho_i/\gamma)$, which linearly penalizes if the fraction of voters deviates below $\gamma \alpha_i$, and does not penalize if the fraction exceeds this amount. In Appendix~\ref{sec:monotone}, we present a $2$-approximation algorithm for arbitrary monotone functions.

\medskip \noindent{\bf Variation Distance.} The objective functions described above give equal importance to satisfying small demographics (small $\alpha$) and large ones (large $\alpha$), since the objective normalizes each demographics' vote fraction by $\alpha$. We could alternatively consider an objective that minimizes the variation distance between the population fractions $\{\alpha_i\}$ and the vote fractions $\{\rho_i \alpha_i\}$, so that $f(\rho_i) = - \alpha_i \left| 1 - \rho_i \right|$. Such a global approach may ignore small demographics if there are a few of them, which has the advantage of potentially reducing the cost per vote.  In Appendix~\ref{sec:tvd}, we present an approximation for the variation distance objective.

\subsection{Parameter Learning} 
\label{sec:online}
In developing the model and algorithms, we assume that for each demographic $i$, the conversion rate $\phi_i$ and the proportionality in spend $q_i$ values are known. In practice, we will have to learn these values. If the platform performs probabilistic pacing, it can reveal the scaling factor $\sum_{i \in S_j} q_i$ for each campaign $j$. If we split the $n$ demographics in a different fashion into the $k$ campaigns over $n/k$ days, then we get $n$ linearly independent equations, from which we can infer the $n$ values $q_i$. 
We can now use the number of conversions $\eta_{ij}$ per dollar  spent on the campaign to infer the conversion rates $\phi_i$ via Eq~(\ref{eq:eta2}). 
We leave designing bandit algorithms to learn these parameters with low regret, or using estimates from another city as prior, as a future research direction.

\subsection{Simulation on Greensboro, NC Data}
\label{sec:simulate}
We simulated our algorithms using data from the Greensboro 2019 Participatory Budgeting Election and the North Carolina voter list to show that had the primary goal of our experiment been to achieve demographic balance, we could have achieved it via suitable campaign design and budget allocation. 

We considered the the $\gamma$-step function objective presented above. The set of demographics $D$ is defined with three dimensions: African American and non-African American race; young, middle, and older ages; and male and female gender. A demographic $i \in D$ is a race-age-gender combination so that $|D| = 12$.   We set proportionality values $q_i$ to the voter list sizes for the corresponding demographics.  In Appendix~\ref{app:budget}, we present experimental evidence that Facebook's allocation mechanism approximately satisfies proportionality for this setting of $q_i$, showing it is a reasonable approximation to make in our simulation. The conversion rates $\phi_i$ were inferred from the custom campaigns (AfrAmC) and (OtherC) (Appendix~\ref{app:greensboro}) targeted toward African Americans and non-African-Americans  respectively. Using the amount of money spent for each age-gender combination reported by Facebook, if $v_i$ is the number of surveys obtained and $b_i$ is the spend by demographic $i$ from the corresponding campaign, then we set the conversion rate $\phi_{i} = \frac{\max(1,v_{i})}{b_{i}}$. 


\small
\begin{table}[htbp] 
\parbox{.5\linewidth}{
\centering
 \begin{tabular}{|l | c | c | c | c |} 
    \hline 
     &$\gamma$=0.5 &0.7 &0.9 &1.0\\   
    \hline  
     Count &11 & 10 & 9 & 9 \\ 
    CPV 
    & \$23.73 & \$32.07 & \$39.30 & \$43.67 \\ 
    Size ratio 
    & 0.438 & 0.336 & 0.202 & 0.202 \\ 
    \hline 
    \end{tabular} 
    \subcaption{$k=2$ Campaigns}
}
\parbox{.5\linewidth}{ %
    \centering
    \begin{tabular}{| c | c | c | c |}
    \hline 
        $\gamma$=0.5 &0.7 &0.9 &1.0\\ 
    \hline 
    12 & 12 & 11 & 11 \\
    \$22.47 & \$25.09 & \$30.62 & \$34.02 \\ 
     0.068 & 0.076  & 0.076 & 0.076 \\
    \hline 
    \end{tabular} 
    \subcaption{$k=4$ Campaigns}
}
\caption{Count is the number of demographics satisfied; CPV is the combined budget of the $k$ campaigns to achieve one vote; and size ratio is the ratio of the size of the smallest voter list over all campaigns to that of the total voter list.} 
\label{table:sim_results}
\end{table}
\normalsize

\noindent {\bf Results.}
In Table \ref{table:sim_results}, we present the results of simulating our algorithm for two different values of $k$, and various settings of $\gamma$. 
We show that a small number, $k=4$, of  campaigns suffice to satisfy $11$ out of the $12$ demographics exactly, or all $12$ demographics to $\gamma \ge 70\%$ of their fair proportions. As expected, the results show an increased CPV for increasing values of $\gamma$, and increasing the number of campaigns results in smaller campaigns.
However, the final row in the table also shows that bounding the number of campaigns is a reasonable approximation to ensuring large audiences. 

We emphasize that this is only a proof of concept simulation. We used fixed settings of $\{q_i, \phi_i\}$, where the former is an approximation and the latter learnt from data. Had we run the resulting campaigns on Facebook,  we could use the following heuristic to periodically adjust  campaigns and budgets: Fix $q_i$ to be equal to voter list size. Periodically recalculate the $\phi_i$ based on actual conversions $\eta_{ij}$ via Eq~(\ref{eq:eta2}), recompute the optimal campaigns and budgets using Theorem~\ref{thm:main}, and update them accordingly.

\section{Discussion and Open Questions}
\label{sec:discuss}
Our work  suggests several directions for future research.  We divide these into four categories.

\medskip \noindent {\bf The Platform Perspective.}  We addressed the problem of achieving demographically balanced outcomes from the advertiser's perspective. A natural question that arises is why not put all the onus on the platform? An obvious answer is that not all advertisers care about such equity at the expense of overall ROI, and further,  we use equity in outcome that a platform cannot reasonably measure. A more nuanced answer is that it is not clear if platforms should even provide explicit support for advertising to a diverse cohort, since a tool for cost-effectively allowing an advertiser to reach specific audiences can also be used for discriminating against protected groups and privacy-loss. It is an interesting open question to study additional tools a platform can provide to enable fairness, and the security and privacy implications of such tools.  

\medskip \noindent {\bf The Advertiser Perspective.} While voter lists can be a part of the solution for recruiting a diverse cohort, they can  exclude certain segments of society, particularly, individuals who have not voted in the past.  Hence, it would be desirable to develop additional mechanisms for reliably reaching a target demographic. We did not find ``lookalike targeting" on Facebook based on race to be very effective for this purpose, as we detailed in Section~\ref{sec:misc} and Appendix~\ref{app:lookalike}. Similarly, we also allocated a minimum portion of our budget to a complementary campaign; however, this again negates much of the effect of creating a balanced cohort. 

It would also be interesting to explore using a combination of online advertising and social seeding (e.g.~\cite{tambe}). While social influence maximization has been well studied (e.g.~\cite{kkt:social}), the focus has not been on getting a diverse cohort.  
Further, an advertiser-centric approach can allow the advertiser to use multiple platforms to achieve the desired diversity when different platforms are more cost-effective or have a bigger reach for different demographics. For instance, the City of Greensboro independently also advertised on Instagram and through YouTube with two targeted campaigns based on household income being over and under the city median.  This resulted in a more specific campaign for African Americans than the campaigns in Section~\ref{sec:pareto}, but less effective (high cost per vote).  Similarly, it would be interesting to check how generalizable our insights are to settings beyond civic engagement or to other fairness criteria.

\medskip \noindent {\bf The Normative Perspective.} In addition to the ethical questions discussed above, our work raises other normative questions that need to be carefully studied. First, it is not obvious what the right measure of diversity should be: should an advertiser aim for equitable amounts of budget spent for different demographics, equitable number of impressions, or equitable number of participants recruited?  These objectives can be in conflict as shown in Section~\ref{sec:pareto}: optimizing for equitable proportions would result in a {\em lower} absolute number of African American votes. Further, how do we determine the right granularity for measuring diversity? Not surprisingly, we observe in our experiments that if we measure diversity separately for the coarse demographics of age and race,  we might actually do worse for gender; we present the result in Appendix~\ref{app:gender}.

\medskip \noindent {\bf Algorithmic Challenges.} The model and algorithms in Section~\ref{sec:algorithms} do not currently have additional constraints such as limits on cost per vote, or limiting campaigns to have some nice structure such as contiguous age groups or geography. Further, since we are optimizing proportions of votes,  even a few days' worth of mis-allocated budget based on inaccurate parameter settings can significantly skew the proportions. Can we devise low-regret learning algorithms for the parameters?

\paragraph{Acknowledgment:} This work is supported by NSF grants CCF-1637397 and IIS-1447554; ONR award N00014-19-1-2268; and an Adobe Digital Experience research award. We thank the cities of Durham and Greensboro, NC for their help and support; Brandon Fain, Nikhil Garg, and Aleksandra Korolova for several useful discussions; and Anilesh Krishnaswamy and Sukolsak Sakshuwong for help with the infrastructure for data collection.

\bibliographystyle{abbrv}
\bibliography{bibliography/refs}

\begin{thebibliography}{10}

\bibitem{agwy:pacing14}
D.~Agarwal, S.~Ghosh, K.~Wei, and S.~You.
\newblock Budget pacing for targeted online advertisements at linkedin.
\newblock {\em Proceedings of the 20th ACM SIGKDD International Conference on
  Knowledge Discovery and Data Mining}, pages 1613--1619, 2014.

\bibitem{probpacing}
D.~Agarwal, S.~Ghosh, K.~Wei, and S.~You.
\newblock Budget pacing for targeted online advertisements at linkedin.
\newblock {\em Proceedings of the ACM SIGKDD International Conference on
  Knowledge Discovery and Data Mining}, 08 2014.

\bibitem{agm:auction06}
G.~Aggarwal, A.~Goel, and R.~Motwani.
\newblock Truthful auctions for pricing search keywords.
\newblock {\em Proceedings of the seventh ACM conference on Electronic
  Commerce}, pages 1--7, June 2006.

\bibitem{Korolova2}
M.~Ali, P.~Sapiezynski, M.~Bogen, A.~Korolova, A.~Mislove, and A.~Rieke.
\newblock Discrimination through optimization: How facebook's ad delivery can
  lead to skewed outcomes.
\newblock {\em CoRR}, abs/1904.02095, 2019.

\bibitem{AtheyE}
S.~Athey and G.~Ellison.
\newblock Position auctions with consumer search.
\newblock Working Paper 15253, National Bureau of Economic Research, August
  2009.

\bibitem{bkmm:pacing17}
S.~Balseiro, A.~Kim, M.~Mahdian, and V.~Mirrokni.
\newblock Budget management strategies in repeated auctions.
\newblock {\em Proceedings of the 26th International Conference on World Wide
  Web}, pages 15--23, 2017.

\bibitem{Balsiero}
S.~R. Balseiro, J.~Feldman, V.~S. Mirrokni, and S.~Muthukrishnan.
\newblock Yield optimization of display advertising with ad exchange.
\newblock {\em Management Science}, 60(12):2886--2907, 2014.

\bibitem{BenadeGP19}
G.~Benad{\`{e}}, P.~G{\"{o}}lz, and A.~D. Procaccia.
\newblock No stratification without representation.
\newblock In {\em Proceedings of the 2019 {ACM} Conference on Economics and
  Computation, {EC} 2019, Phoenix, AZ, USA, June 24-28, 2019.}, pages 281--314,
  2019.

\bibitem{Borgs}
C.~Borgs, J.~Chayes, N.~Immorlica, K.~Jain, O.~Etesami, and M.~Mahdian.
\newblock Dynamics of bid optimization in online advertisement auctions.
\newblock In {\em Proceedings of the 16th International Conference on World
  Wide Web}, WWW '07, pages 531--540, 2007.

\bibitem{CelisV}
L.~E. Celis, A.~Mehrotra, and N.~K. Vishnoi.
\newblock Toward controlling discrimination in online ad auctions.
\newblock In {\em Proceedings of the 36th International Conference on Machine
  Learning, {ICML}}, pages 4456--4465, 2019.

\bibitem{dg:fair18}
S.~Corbett-Davies and S.~Goel.
\newblock The measure and mismeasure of fairness: A critical review of fair
  machine learning.
\newblock 2018.

\bibitem{Devanur}
N.~R. Devanur and T.~P. Hayes.
\newblock The adwords problem: Online keyword matching with budgeted bidders
  under random permutations.
\newblock In {\em Proceedings of the 10th ACM Conference on Electronic
  Commerce}, EC '09, pages 71--78, 2009.

\bibitem{Devanur2}
N.~R. Devanur, K.~Jain, B.~Sivan, and C.~A. Wilkens.
\newblock Near optimal online algorithms and fast approximation algorithms for
  resource allocation problems.
\newblock In {\em Proc. $12^{th}$ ACM EC}, pages 29--38, 2011.

\bibitem{DworkI}
C.~Dwork and C.~Ilvento.
\newblock Fairness under composition.
\newblock In {\em 10th Innovations in Theoretical Computer Science Conference,
  {ITCS} 2019, January 10-12, 2019, San Diego, California, {USA}}, pages
  33:1--33:20, 2019.

\bibitem{EOS07}
B.~Edelman, M.~Ostrovsky, and M.~Schwarz.
\newblock Internet advertising and the {G}eneralized {S}econd-{P}rice auction:
  Selling billions of dollars worth of keywords.
\newblock {\em American Economic Review}, 97(1):242--259, 2007.

\bibitem{ekksk:ads18}
M.~Eslami, S.~R.~K. Kumaran, C.~Sandvig, and K.~Karahalios.
\newblock Communicating algorithmic process in online behavioral advertising.
\newblock {\em Proceedings of the ACM CHI Conference on Human Factors in
  Computing Systems}, 2018.

\bibitem{GoelLMNP}
G.~Goel, S.~Leonardi, V.~S. Mirrokni, A.~Nikzad, and R.~P. Leme.
\newblock Reservation exchange markets for internet advertising.
\newblock In {\em 43rd Intl. Colloq. Automata, Languages, and Programming,
  {ICALP} 2016}, pages 142:1--142:13, 2016.

\bibitem{gs:rds10}
S.~Goel and M.~J. Salganik.
\newblock Assessing respondent-driven sampling.
\newblock {\em Proc Natl Acad Sci}, 107(15):6743--6747, 2010.

\bibitem{gkmp:liquid18}
P.~G{\"{o}}lz, A.~Kahng, S.~Mackenzie, and A.~D. Procaccia.
\newblock The fluid mechanics of liquid democracy.
\newblock {\em Proceedings of the 14th International Web and Internet Economics
  Conference}, pages 188--202, 2018.

\bibitem{hassin1991improved}
R.~Hassin and A.~Tamir.
\newblock Improved complexity bounds for location problems on the real line.
\newblock {\em Operations Research Letters}, 10(7):395--402, 1991.

\bibitem{kkt:social}
D.~Kempe, J.~Kleinberg, and E.~Tardos.
\newblock Maximizing the spread of influence through a social network.
\newblock {\em Theory of Computing}, 11(4):105--147, 2015.

\bibitem{kkr:pif2020}
M.~P. Kim, A.~Korolova, G.~N. Rothblum, and G.~Yona.
\newblock Preference-informed fairness.
\newblock {\em Proceedings of the 11th Conference on Innovations in Theoretical
  Computer Science}, 2020.

\bibitem{Lappas2009}
T.~Lappas, K.~Liu, and E.~Terzi.
\newblock Finding a team of experts in social networks.
\newblock In {\em Proceedings of the 15th ACM SIGKDD International Conference
  on Knowledge Discovery and Data Mining}, KDD '09, pages 467--476, 2009.

\bibitem{LiShan2010}
C.~{Li} and M.~{Shan}.
\newblock Team formation for generalized tasks in expertise social networks.
\newblock In {\em 2010 IEEE Second International Conference on Social
  Computing}, pages 9--16, Aug 2010.

\bibitem{Liemhetcharat2014}
S.~Liemhetcharat and M.~Veloso.
\newblock Weighted synergy graphs for effective team formation with
  heterogeneous ad hoc agents.
\newblock {\em Artificial Intelligence}, 208:41 -- 65, 2014.

\bibitem{Mehta}
A.~Mehta, A.~Saberi, U.~Vazirani, and V.~Vazirani.
\newblock Adwords and generalized online matching.
\newblock {\em J. ACM}, 54(5), 2007.

\bibitem{MITMORAL}
The {MIT} moral machine.
\newblock \url{http://moralmachine.mit.edu}.

\bibitem{myerson}
R.~B. Myerson.
\newblock Optimal auction design.
\newblock {\em Mathematics of Operations Research}, 6(1):58--73, 1981.

\bibitem{NasrT}
M.~Nasr and M.~Tschantz.
\newblock Bidding strategies with gender nondiscrimination: Constraints for
  online ad auctions, 2019.

\bibitem{PBP}
The participatory budgeting project.

\bibitem{Rahman2019}
H.~Rahman, S.~B. Roy, S.~Thirumuruganathan, S.~Amer-Yahia, and G.~Das.
\newblock Optimized group formation for solving collaborative tasks.
\newblock {\em The VLDB Journal}, 28(1):1--23, Feb. 2019.

\bibitem{sandholm99}
T.~Sandholm, K.~Larson, M.~Andersson, O.~Shehory, and F.~Tohm{\'{e}}.
\newblock Coalition structure generation with worst case guarantees.
\newblock {\em Artif. Intell.}, 111(1-2):209--238, 1999.

\bibitem{Parkes}
B.~Ustun, Y.~Liu, and D.~Parkes.
\newblock Fairness without harm: Decoupled classifiers with preference
  guarantees.
\newblock In K.~Chaudhuri and R.~Salakhutdinov, editors, {\em Proceedings of
  the 36th International Conference on Machine Learning}, volume~97 of {\em
  Proceedings of Machine Learning Research}, pages 6373--6382, Long Beach,
  California, USA, 09--15 Jun 2019. PMLR.

\bibitem{tambe}
A.~Yadav, L.~Soriano~Marcolino, E.~Rice, R.~Petering, H.~Winetrobe, H.~Rhoades,
  M.~Tambe, and H.~Carmichael.
\newblock Psinet: Assisting hiv prevention among homeless youth by planning
  ahead.
\newblock {\em Ai Magazine}, 37:47--62, 2016.

\end{thebibliography}

\appendix
\section{Details of Experiments}
\label{app:details}
In this section, we present the exact campaigns we created in Durham and Greensboro, and the rationale for doing so.  

In our implementations, we considered three age groups (young, middle, and old), two racial categories (African American, non-African American), two ethnic categories (Hispanic, non-Hispanic), and two genders (male, female) as the possible demographics on which we seek balance. In our experiments, the parallel campaigns differed in the targeted audience and allocated budget. We considered various campaigns described in our narrative below. 

For Durham, the total amount spent was \$ 8,617.11 resulting in 727,170 total impressions and a total of 4920 English language surveys was collected over 31 days (surveys include surveys resulting from organic participation). For Greensboro, the total amount spent was \$ 9,681.56 resulting in 1,015,204 impressions and a total of 3,584 English language surveys was collected (including from organic participation).

Our campaigns respected some ethical boundaries: We did not use options that shared private data with third parties, for instance by allowing the advertising platform to track conversion until the survey ('pixel') or by creating a lookalike audience based on completed surveys or verification data. We also did not acquire databases with personal data that are not part of the public record.

\subsection{Experiment in Durham NC: Facebook-inferred Targeting}
\label{app:durham}
We ran the following sets of parallel campaigns in Durham. For each campaign, we  required that targeted individuals were located within the city (i.e. they were eligible to participate) and 18 years or older (for consent reasons).
\begin{itemize}
    \item[D1] 4 campaigns using Facebook-inferred characteristics. A campaign targeting the general population (General), interest in Hispanic culture and Hispanic Multicultural Affinity (Hisp1), interest in African-American culture and African-American Multicultural Affinity (AfrAm1) and the general population excluding everyone with any college experience (Educ).
    \item[D2] 2 campaigns (Hisp2 and AfrAm2) using Facebook-inferred characteristics, similar to D1 but only targeting multicultural affinity.
    \item[D3] 5 campaigns. Besides Hisp1 and AfrAm1, also two custom audiences (HispC and AfrAmC) created from Hispanic/Latino and African American voter lists and one 'lookalike' audience (AfrAmLAL) that Facebook created based on AfrAmC. 
\end{itemize}

We started in Durham naively with the targeting parameters (D1) as they were provided by Facebook. As described in Section~\ref{sec:pareto}, we quickly learned that all three targeted campaigns were quite imprecise with regards to their actual outcome. On the suggestion of a Facebook support agent, we tried targeting only by multicultural affinity (D2). As described in Appendix~\ref{app:afram2}, this was also very ineffective. This led us to create custom audiences using voter lists (D3). Finally, we added a lookalike audience based on one of the custom audiences, to see if we could use the voter lists to effectively target people of the same demographics that never voted in an election. These results for lookalike audiences are described in Section~\ref{sec:greensboro}.

\subsection{Experiment in Greensboro NC: Custom Audiences using Voter Lists}
\label{app:greensboro}
We used the following sets of parallel campaigns in Greensboro. As before, for each campaign, we required that targeted individuals were located within the city (i.e. they were eligible to participate) and 18 years or older (for consent reasons).

\begin{itemize}
    \item[G1] 6 campaigns, 5 of which were custom using overlapping subsets of voter lists targeting age 18-44 (YoungC), age 45-65 (MiddleC), age 65+ (OldC), African-American voters (AfrAmC), all voters except African-American voters (OtherC) and finally a complement targeting all people in the city except those already matched in G1 (Complement1).
    \item[G2] The same as G1, but (AfrAmC) was randomly split into 4 non-overlapping custom audiences (AfrAmC-1, AfrAmC-2, AfrAmC-3 and AfrAmC-4).
    \item[G3] 8 campaigns, 5 of which (YoungLAL, MiddleLAL, OldLAL, AfrAmLAL, OtherLAL) were lookalike audiences based on G1, a lookalike audience based on an audience of all female voters (FemaleLAL) and all male voters (MaleLAL) and the complement of all other G3 lookalike audiences (Complement3). Lookalike audiences were created by Facebook with a 5\% (age) or 10\% (race/gender) target, resulting in a national audience that is then filtered by location (Greensboro) and age (18+). 
    \item[G4] 6 campaigns, using Facebook-inferred characteristics targeting age 18-44 (YoungF), age 45-65 (MiddleF), age 65+ (OldF), interest in African-American culture and African-American multicultural affinity (AfrAmF), female users (FemaleF) and male users (MaleF). 
\end{itemize}

We started with (G1) and included the ``complement" campaign  targeting the whole city except those people part of one of the other campaigns, in order not to fully exclude them from the advertising process. These results are presented in Section~\ref{sec:greensboro}. In order to test a hypothesis that targeting smaller audiences is more expensive (see Appendix~\ref{sec:k_campaigns}), we continued in (G2) with the same campaigns, but split one of the audiences randomly into four audiences of equal size. 

In (G3), we explored the effectiveness in targeting by lookalike audiences along different demographic axes, by creating custom audiences sliced by age (3), race (2) and gender (2) and again one complement campaign; these results are presented in Appendix~\ref{app:lookalike}. We also further explored the effectiveness of targeting by Facebook-inferred characteristics in (G4). In (G4), non-African Americans could not be targeted because Facebook does not allow to exclude people African American multicultural affinity from a campaign. Our preliminary analysis indicates that targeting by proxies for race is as ineffective as in Durham; however, there are too few votes to assess confidence and we do not report these results. 

\section{Additional Empirical Analysis}
\label{app:additional}
We now provide additional empirical results that we allude to in Sections~\ref{sec:targeting},~\ref{sec:algorithms}, and~\ref{sec:discuss}.

\subsection{Performance of AfrAm2 Campaign in Durham}
\label{app:afram2}
In set (D2), when ran one ad targeting African American multicultural affinity. Assuming a Poisson distribution for African-American votes generated, the ad had an average of $0.002$ African American votes per dollar, with 95\% confidence interval $[0,0.013]$. The specificity was $0.09$, with a $95\%$ confidence interval of $[0,0.37]$. This shows a performance comparable to (AfrAm1), and is Pareto-dominated on African-American votes per dollar by (General).

\subsection{Lookalike Audiences in Greensboro} 
\label{app:lookalike}
The ``lookalike" campaigns (G3) in Greensboro NC didn't provide enough votes to analyze their specificity with surveys. Because Facebook reports detailed statistics for age and gender, we can evaluate the proportion of the budget, views and clicks that belongs to each demographic. We observe that while we inform the lookalike campaigns with very different (e.g. all male and all female) audiences, the spending and view proportions across gender (male) and age (young) are indistinguishable (see Table~\ref{table:gso_lookalike}). This can only be explained if the created audience includes sufficient people outside the targeted demographic. 

We summarize these observations below.

\begin{observation}
Campaigns created using the ``lookalike" feature on Facebook even when seeded with lists specific to age and gender demographics, result in exposure and actions that are indistinguishable across those demographics.
\end{observation}

\begin{table}[!htb]  
\centering
 \begin{tabular}{|| c || c | c || c | c | c ||} 
 \hline
 Proportion of: & \multicolumn{2}{c||}{Male} & \multicolumn{3}{c||}{Young} \\
 \hline
 Targeting Criterion: & Female & Male & Young & Middle & Old \\  
 \hline
 Custom Audience & 0.000 & 1.000 & 1.000 & 0.000 & 0.000 \\
 \hline
 Spend & 0.544 & 0.558 & 0.789 & 0.781 & 0.679 \\ 
 \hline
 Impressions & 0.555 & 0.528 & 0.822 & 0.835 & 0.770 \\
 \hline
 Clicks & 0.655 & 0.589 & 0.835 & 0.820 & 0.607 \\
 \hline
\end{tabular}
\caption{Specificity of complementary Lookalike Campaigns; the voter list demographic that is used for creating the respective lookalike campaign is listed in the second row. For each campaign, we find the fraction of spend, impressions, and clicks that Facebook allocates to male and young audiences.} 
\label{table:gso_lookalike}
\end{table} 

\begin{figure}[htbp]
\begin{center}
\parbox{.46\linewidth}{
\includegraphics[width = 0.45\textwidth]{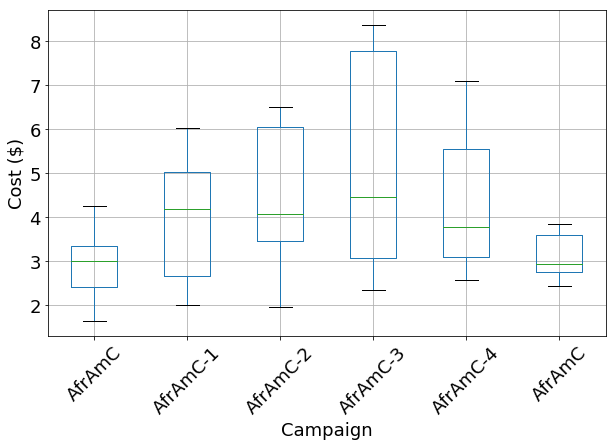}
\subcaption{\label{fig:gso_aa_daily_cpc} Box plots of the daily cost per click for Greensboro African American custom campaigns, where the first and last campaigns are Full campaigns and the rest are run with a fourth of the voter list each.}
}
\hspace{0.2in}
\parbox{.46\linewidth}{
\includegraphics[width = 0.45\textwidth]{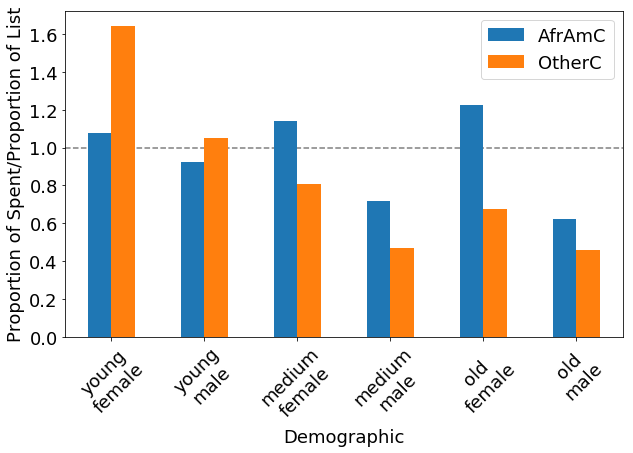}
\subcaption{\label{fig:gso_prop_spending_w_race} The ratio of the proportion of the budget allocated to the proportion of the voter list for that demographic, each with respect to the overall campaign.}
}
\end{center}
\caption{Effect of campaign size (left) and proportionality of budget allocation (right)} 
\end{figure}

\subsection{Large versus Small Audience Sizes}
\label{sec:k_campaigns}
We tested whether reducing audience size had an impact on the performance of advertising campaigns. Towards this end, we split the voter list for African Americans in Greensboro (the campaign with the highest allocated budget) into four equal parts. We first ran a larger campaign targeted to the full list of African American voters with a daily budget of around $\$100$. We then stopped this campaign and ran four smaller campaigns, each with one equal sized split of the voter list and a daily budget of $\$25$. The results of the experiment are shown in Fig. \ref{fig:gso_aa_daily_cpc}, and leads to the following observation. 

\begin{observation}
The smaller split campaigns (with smaller audiences) each has a greater range in cost per click, a higher maximum cost per click, and higher mean cost per click than the campaigns with the full voter list (with larger audience).
\end{observation}

In order to rule out the effects of ad fatigue and exhaustion of the voter list as the cause, we again ran the full campaign for another six days. Notice that the full campaign has comparable performance in both the initial and final time segments, showing that the effect we observe cannot be explained by ad fatigue.

\subsection{Relation of Spend to Proportion of Voter List}
\label{app:budget}
We now analyze how Facebook splits budget across sub-demographics within a voter list. As mentioned before, when an advertiser submits a list of people to Facebook for a campaign audience, Facebook uses the information in the list to identify corresponding Facebook profiles. The budget set for the campaign is then spent advertising directly to the list of matched Facebook profiles. 

We analyze the Greensboro (AfrAmC) and (OtherC) campaigns described above. For each campaign, we consider the sub-demographics induced by combinations of age and gender. For each of these sub-demographics, Facebook shows us the spend per campaign. We calculate the ratio of the fraction of the spend for that sub-demographic to the fraction of the voter list that is of that sub-demographic.  In Fig.~\ref{fig:gso_prop_spending_w_race}, we show the actual ratios lie between at least $0.4$ and $1.60$. This shows that to a reasonable approximation, for each demographic $i$,  we can assume $q_i$ is proportional to the size of the corresponding voter list. Of course, the resulting conversions across demographics within this list could be highly unequal (see Section~\ref{sec:greensboro}). Our goal is not to argue that Facebook's allocation mechanisms are perfectly proportional, but instead to come up with a reasonable approximation that we can use in our simulation and experiments. 
 
\subsection{Outcomes Orthogonal to Targeting Criteria}
\label{app:gender}
The custom campaigns in Greensboro did not explicitly target gender and a priori we expect the gender ratios to be proportional. We show in Table~\ref{table:gso_women} that the resulting votes have a disproportionate fraction of females, while they make up approximately half of each list. Therefore, if we impose equity of effort via custom campaigns along the dimensions of race and age, this has the side-effect of producing inequitable outcomes along the gender dimension.  The take-away is that if we seek parity along a dimension (such as gender), we need to explicitly consider it in our optimization models and algorithms. This motivates our overall problem of balancing fine grained demographics.

\begin{table}[htbp]
\centering
 \begin{tabular}{||c | c | c | c | c | c | c ||} 
 \hline
& YoungC & AfrAmC & MediumC & OldC & OtherC \\ 
 \hline
 Fraction of Votes & \textbf{0.49} & \textbf{0.74} & \textbf{0.75} & \textbf{0.83} & \textbf{0.86} \\ 
 \hline
 Fraction of List & 0.59 & 0.61 & 0.54 & 0.58 & 0.55 \\ 
 \hline
\end{tabular}
\caption{\label{table:gso_women}Fraction of female votes to overall votes for Greensboro 2019 Custom Audience Campaigns.} 
\end{table}

\section{Algorithms for Other Objective Functions}
\label{app:model}
\subsection{Algorithm for Monotone Functions} 
\label{sec:monotone}
We first consider the case of monotone functions, where the objective can be written as
$$ \mbox{Minimize} \sum_{i \in D} f\left( \frac{Y_{\sigma(i)}}{\beta_i} \right) $$
subject to Constraint (\ref{eq:atmostone}). Here, $f$ is an arbitrary monotonically non-decreasing function with $f(0) = 0$, and $f(\rho) = 1$ for $\rho \ge 1$. 

Let $OPT$ denote the optimal objective value, as well as the optimal solution. The first thing to note is that Claim~\ref{claim:ahead} need not hold: A demographic $i \in D$ can now potentially be assigned to campaign $j$ with $Y_j$ much smaller than $\beta_i$. This complicates the design of the dynamic program. We now show the following theorem:

\begin{theorem}
\label{thm:monotone}
A solution of objective value at least $OPT/2$ can be computed in time $O(k n^5)$.
\end{theorem}

In order to prove the theorem, we first discuss the structure of $OPT$. Given any solution with $Y_1 \le Y_2 \le \cdots \le Y_k$, suppose $i \in D$ satisfies $Y_{j} \le \beta_i \le Y_{j+1}$. Then this $i$ is assigned in one of three possible ways:

\begin{enumerate}
    \item It is not assigned to any campaign, so that $\sigma(i) = 0$; or
    \item It is assigned to campaign $j' > j$, contributes $1$ to the objective, and has expected number of conversions $w_i Y_{j'}$. In this case, we can assume $j' = j+1$ since it reduces the expected number of conversions while preserving the objective; or
    \item It is assigned to campaign $j' \le j$, so that it contributes $f\left( \frac{Y_{j'}}{\beta_i}\right)$ to the objective, and has expected number of conversions $w_i Y_{j'} \le w_i \beta_i$. 
\end{enumerate}

Note now that at least half the objective is contributed by Cases (1) and (2), or by Cases (1) and (3). We therefore solve for the best solution where each demographic is constrained to be assigned as in Cases (1) and (2), and we then solve for the best solution where each demographic is constrained to be assigned as in Cases (1) and (3). We then take the better of the two solutions. This will yield a $2$-approximation to $OPT$, proving Theorem~\ref{thm:monotone}.

\paragraph{Optimal Solution for Case (2).} We first consider the case where each demographic $i$ is either assigned to the smallest $j$ such that $Y_j \ge \beta_i$ or to campaign $0$. However, this case is exactly the same as the formulation in Section~\ref{sec:step} when $\gamma = 1$. By Theorem~\ref{thm:main}, the optimal solution for this case can be computed in $O(k n^5)$ time.

\paragraph{Optimal Solution for Case (3).} In this case, note that since each demographic $i \in D$ is assigned to a campaign $\sigma(i) = j$ with $Y_j \le \beta$, the expected number of conversions is
$$ \sum_{i \in D} w_i Y_{\sigma(i)} \le \sum_{i \in D} w_i \beta_i = \sum_{i \in D} \alpha_i = 1$$

Therefore, Constraint~(\ref{eq:atmostone}) is satisfied {\em regardless} of how $i \in D$ is assigned. In order to maximize the objective, each $i \in D$ should therefore be assigned to the largest $Y_j$ such that $Y_j \le \beta_i$. 

We can find this solution by a simple dynamic program. Let $W_j(i)$ denote the best set of $j$ campaigns, where $Y_j = \beta_i$. The final goal is find $\max_i W_k(i)$. For $j > 1$, we have the recurrence:
$$ W_j(i) = \max_{1 \le \ell \le i} \left\{ W_{j-1}(\ell) + \sum_{q=\ell}^{i-1} f(\beta_{\ell}/\beta_q) \right\}$$ 
which corresponds to assigning all demographics $q$ with $\ell \le q < i$ to $Y_{j-1} = \beta_{\ell}$, and taking the best possibility. The running time of this dynamic program is $O(k n^2)$.

Taking the better of the solutions to the two dynamic programs yields value at least $OPT/2$ while preserving Constraint~(\ref{eq:atmostone}), completing the proof of Theorem~\ref{thm:monotone}.

\subsection{Minimizing Total Variation Distance}
\label{sec:tvd}
In this section, we treat demographic fractions $\{\alpha_i\}$ and the vote fractions $\{\rho \alpha\}$ as distributions and minimize the variation distance between them. Recall that $\sum_{i \in D} \alpha_i = 1$, and that $f(\rho_i) = - \alpha_i \| 1 - \rho_i\|$, so that Objective~(\ref{eq:obj}) becomes:
$$ \mbox{Minimize} \sum_{i \in D} \alpha_i \| 1 - \rho_i\|$$
As mentioned before, such a global approach may ignore small demographics if there are not too many of them, which has the advantage of potentially reducing the cost per vote. It also leads to a more efficient dynamic program compared to the ones presented in Section~\ref{sec:step} and Appendix~\ref{sec:monotone}.

Formally, let $Y_0 = 0$ and $S_0$ be the corresponding demographics not assigned to any campaign. Recall that the fraction of conversions for demographic $i$ can be written as $w_i Y_{\sigma(i)}$ if $i \in S_j$. Similarly, $\alpha_i = w_i \beta_i$. The variation distance objective function is therefore:
$$ \mbox{Minimize} \sum_{i \in D} w_i \left| Y_{\sigma(i)} - \beta_i \right|$$
subject to the constraint that the expected number of conversions is $1$, that is, Constraint~(\ref{eq:exactlyone}). In this setting, it is not clear how to replace Constraint~(\ref{eq:exactlyone}) by Constraint~(\ref{eq:atmostone}).
We now show the following theorem; the non-trivial part of the proof is satisfying Constraint~(\ref{eq:exactlyone}). 
\begin{theorem}
\label{thm:tvd}
A solution with TVD at most $2OPT$ and that uses at most $k+1$ campaigns can be computed in $O(kn)$ time.
\end{theorem} 

\paragraph{Weighted $k$-Median Relaxation.} Suppose we consider the relaxed problem where we ignore Constraint~(\ref{eq:exactlyone}). Clearly, the objective value does not increase when we remove a constraint. This relaxation is exactly the weighted $k$-median problem on a line, where the point $i$ is located at $\beta_i$, its weight is $w_i$, and the centers are located at $\{Y_j\}$, with an extra $(k+1)^{st}$ center fixed at $0$. This has a $O(k n)$ time exact algorithm~\cite{hassin1991improved}. The resulting solution yields the $\{Y_j\}$ as well as the associated sets $S_j$. Given this optimal solution $0 = Y_0 \le Y_1 \le \cdots \le Y_k$, each $i \in D$ with $Y_j \le \beta_i \le Y_{j+1}$ is assigned to the closest $Y$, which means it is assigned either to $Y_j$ or to $Y_{j+1}$.

\medskip
Given a solution to this relaxation (which has objective at most $OPT$), we show that Constraint~(\ref{eq:exactlyone}) can be satisfied by increasing the objective by a factor of $2$. Suppose the relaxed optimum has $0 = Y_0 \le Y_1 \le Y_2 \le \cdots \le Y_k$. Let $S_j$ denote the set of demographics assigned to campaign $j$; note that each $i \in S_j$ is closer to $Y_j$ than to any other $Y_{j'}$. For $j \ge 0$, we modify $Y_j$ as follows: Let 

\begin{equation}
    \label{eq:modify}
    \tilde{Y_j} = \left\{ y \ \middle|\ \sum_{i \in S_j, \beta_i \le y} w_i(y - \beta_i) = \sum_{i \in S_j, \beta_i \ge y} w_i(\beta_i - y) \right \}
\end{equation}

The key lemma below shows that the modification satisfies Constraint~(\ref{eq:exactlyone}) and increases the objective by a factor of $2$. Since the above process may make $\tilde{Y_0} > 0$, this results in an extra campaign in the solution, so it will have $k+1$ campaigns. This will complete the proof of Theorem~\ref{thm:tvd}.

\begin{lemma} The modified solution $\tilde{Y_1} \le \cdots \le \tilde{Y_k}$ with sets $S_j$ for campaign $j$ satisfies:
\begin{enumerate}
    \item The TVD satisfies: $\sum_{j=0}^k \sum_{i \in S_j} w_i \left\| \tilde{Y_j} - \beta_i\right\| \le 2 OPT$; and 
    \item The expected number of conversions satisfies: $\sum_{j=0}^k \sum_{i \in S_j} w_i \tilde{Y_j} = 1$.
\end{enumerate}
\end{lemma}
\begin{proof}
Given the $\{Y_j\}$ and sets $S_j$, the contribution of campaign $j$ to the original relaxed objective is $\sum_{i \in S_j} |Y_j - \beta_i|$, and to the modified solution is $\sum_{i \in S_j} |\tilde{Y_j} - \beta_i|$. 

Suppose $\tilde{Y_j} \le Y_j$; the other case is symmetric. Let $A_j = \{i \in S_j, \beta_i \le \tilde{Y_j}\}$. Then, 
$$ \sum_{i \in S_j} |\tilde{Y_j} - \beta_i|  =  2 \sum_{i \in A_j} \left(\tilde{Y_j} - \beta_i\right)   \le 2 \sum_{i \in A_j} \left(Y_j - \beta_i\right)   \le 2 \sum_{i \in S_j} \left|Y_j - \beta_i\right| $$
Here, the first equality follows from Eq~(\ref{eq:modify}); the first inequality follows since $\tilde{Y_j} \le Y_j$, and the final inequality follows since $A_j \subset S_j$. Summing over $j = 0,1,\ldots,k$ proves the first part of the lemma.

Note now that $\sum_{i \in D} w_i \beta_i = \sum_{i \in D} \alpha_i = 1.$ For each $j \ge 0$, Equation~(\ref{eq:modify}) implies that
$ \sum_{i \in S_j} w_i \tilde{Y_j} = \sum_{i \in S_j} w_i \beta_i$.
Summing over all $j$ proves the second part of the lemma.
\end{proof}

\end{document}